\documentclass[a4paper,12pt]{article}
\usepackage[utf8]{inputenc}
\usepackage[english]{babel}
\usepackage{amsmath,amssymb,amsfonts,amsthm,textcomp,natbib}
\usepackage{graphicx}
\graphicspath{{./images/}}
\usepackage{listings}
\usepackage{tabularx}
\usepackage{booktabs}
\usepackage{fullpage}
\usepackage{enumitem}
\newtheorem{theorem}{Theorem}
\newtheorem{definition}{Definition}

\newtheorem{lemma}{Lemma}

\makeatletter
\@addtoreset{proofpart}{theorem}
\makeatother
\usepackage{tikz}
\usetikzlibrary{positioning,arrows,shapes,shadows,snakes}
\usetikzlibrary{intersections,decorations.markings,calc,positioning}

\usepackage{soul}

\makeatletter
\newcommand\RedeclareMathOperator{%
  \@ifstar{\def\rmo@s{m}\rmo@redeclare}{\def\rmo@s{o}\rmo@redeclare}%
}
\newcommand\rmo@redeclare[2]{%
  \begingroup \escapechar\m@ne\xdef\@gtempa{{\string#1}}\endgroup
  \expandafter\@ifundefined\@gtempa
     {\@latex@error{\noexpand#1undefined}\@ehc}%
     \relax
  \expandafter\rmo@declmathop\rmo@s{#1}{#2}}
\newcommand\rmo@declmathop[3]{%
  \DeclareRobustCommand{#2}{\qopname\newmcodes@#1{#3}}%
}
\@onlypreamble\RedeclareMathOperator
\makeatother

\DeclareMathOperator{\cgeq}{\succcurlyeq}

\DeclareMathOperator{\cgt}{\succ}
\DeclareMathOperator{\clt}{\prec}
\DeclareMathOperator{\R}{\mathbf{R}}
\RedeclareMathOperator{\S}{\mathbf{S}}
\DeclareMathOperator{\E}{\mathbf{E}}

\usepackage{breqn}

\begin{document}
\title{Choquet integral in decision analysis – lessons from the axiomatization.}
\author{Mikhail Timonin}
\maketitle

\begin{abstract}
In \citep{timonin2016axiomatization} we developed a general axiomatic treatment of a popular multicriteria decision model - the Choquet
integral. This paper contains extensions of our results to the particular interesting special cases of the Choquet integral, analysis of some
aspects of the Choquet integral model learning, and a discussion of the applications of our results in decision theory.  
\end{abstract}




\section{Introduction}
\label{sec:c4-introduction}

In \citep{timonin2016axiomatization} we developed a general axiomatic treatment of a popular multicriteria decision model - the Choquet
integral. This paper contains extensions of our results to the particular interesting special cases of the Choquet integral, analysis of
some aspects of the Choquet integral model learning, and a discussion of the applications of our results in decision theory.  The Choquet
integral is a powerful aggregation operator which lists many well-known models as its special cases. In this paper we look at these
special cases and provide their axiomatic analysis. In cases where an axiomatization has been previously given in the literature, we connect
the existing results with the framework that we have developed.

Next we turn to the question of learning, which is especially important for the practical applications of the model. So far, learning of the
Choquet integral has been mostly confined to the learning of the capacity. Such an approach requires making a powerful assumption that all
dimensions (e.g. criteria) are evaluated on the same scale, which is rarely justified in practice. Too often  categorical data is
 given arbitrary numerical labels (e.g. AHP), and numerical data is considered cardinally and ordinally commensurate, sometimes after a
simple normalization. Such approaches clearly lack  scientific rigour, and yet they are commonly seen in all kinds of applications. We
discuss the pros and cons of making such an assumption and look at the consequences which our uniqueness results have for the learning
problems.

Finally, we revisit some of the applications we discussed in the Introduction. Apart from MCDA, which is the main area of interest for our
results, we also discuss how the model can be interpreted in the social choice context. We look in detail at the state-dependent utility,
and show how comonotonicity, central to the previous axiomatizations, actually implies state-independency in the Choquet integral model. We
also discuss the conditions required to have a meaningful state-dependent utility representation and show the novelty of our results
compared to the previous methods of building state-dependent models.


\section{Extensions}
\label{sec:spec-cases-choq}

\subsection{Ordinal models}
\label{sec:ordinal-models}

Notable ordinal special cases of the Choquet integral are:
\begin{itemize}
\item Min/Max
\item Order statistic ($k$-smallest element) $OS_k$
\item Lattice polynomial $p^{{\cal A}{\cal B}}$.
\end{itemize}
Moreover, Min/Max are special cases of $OS_k$ ($k=1$ and $k=n$ correspondingly), and $OS_k$ is a special case of the lattice polynomial model, as becomes
evident from the following definitions.
\begin{definition}
  $\cgeq$ can be represented by \emph{MIN}, if exist value functions $\phi_i:X_i \rightarrow \mathbb{R}$ such that for all $x,y
  \in X$ we have
  \begin{equation}
    x \cgeq y \iff \bigwedge_{i \in N}\phi_i(x_i) \geq \bigwedge_{i \in N}\phi_i(y_i),
  \end{equation}
where $\bigwedge$ means minimum.
\end{definition}

\begin{definition}
  $\cgeq$ can be represented by \emph{MAX}, if exist value functions $\phi_i:X_i \rightarrow \mathbb{R}$ such that for all $x,y
  \in X$ we have
  \begin{equation}
    x \cgeq y \iff \bigvee_{i \in N}\phi_i(x_i) \geq \bigvee_{i \in N}\phi_i(y_i),
  \end{equation}
where $\bigvee$ means maximum.
\end{definition}

\begin{definition}
  $\cgeq$ can be represented by an \emph{order statistic} $OS_k$, if exist value functions $\phi_i:X_i \rightarrow \mathbb{R}$ such that for all $x,y
  \in X$ we have
  \begin{equation}
    x \cgeq y \iff \phi_{(k)}(x_{(k)}) \geq \phi_{(k)}(y_{(k)}),
  \end{equation}
where $\phi_{(k)}(z_{(k)})$ stands for $k$th smallest element of $(\phi_1(z_1),\ldots,\phi_n(z_n))$.
\end{definition}
An order statistic can be written in a CNF and DNF-like\footnote{Conjunctive normal form and disjunctive normal form.} forms
(e.g. \citealp{ovchinnikov1996means}):
\begin{equation}
  OS_k = \bigwedge_{K \subset N \atop \vert K \vert =k}\bigvee_{i \in K} \phi_i(x_i) = \bigvee_{K \subset N \atop \vert K \vert = n - k + 1}\bigwedge_{i \in K} \phi_i(x_i).
\end{equation}
Obviously, MIN and MAX are particular cases of $OS_k$ with $k=1$ and $k=n$ correspondingly.

\begin{definition}
  $\cgeq$ can be represented by a \emph{lattice polynomial} $p^{{\cal A}{\cal B}}$, if exist value functions
  $\phi_i:X_i \rightarrow \mathbb{R}$ such that for all $x,y \in X$ we have
  \begin{equation}
    x \cgeq y \iff p^{{\cal A}{\cal B}}(\phi_1(x_1),\ldots,\phi_n(x_n)) \geq p^{{\cal A}{\cal B}}(\phi_1(y_1),\ldots,\phi_n(y_n)),
  \end{equation}
where $p^{{\cal A}{\cal B}}(\phi_1(z_1),\ldots,\phi_n(z_n))$ is an expression which includes elements of $(\phi_1(z_1),\ldots,\phi_n(z_n))$ and symbols $\vee$
and $\wedge$.
\end{definition}
We can write any lattice polynomial in DNF and CNF as well:
\begin{equation}
  p^{{\cal A}{\cal B}}(\phi_1(z_1),\ldots,\phi_n(z_n)) = \bigwedge_{K \subset {\cal A}} \bigvee_{i \in K} \phi_i(x_i) = \bigvee_{M \subset {\cal B}} \bigwedge_{i \in M} \phi_i(x_i),,
\end{equation}
where ${\cal A} \subset 2^N$ and ${\cal B} \subset 2^N$ are some collection of subsets of $N$. Obviously, order statistic, hence MIN and
MAX are special cases of an order polynomial.

The following result states that all aforementioned models are special cases of the Choquet integral.
\begin{theorem}[\citealp{murofushi1993some}]
  \label{theo:choquet-eqiv-LP}
  The Choquet integral with respect to a capacity $\nu$ is a lattice polynomial function if and only if $\nu$ is a 0--1 capacity (i.e. only takes
  values 0 or 1). Moreover, any lattice polynomial function on $\mathbb{R}$ is a Choquet integral with respect to a 0--1 capacity.
\end{theorem}

\subsection{Previous characterizations of the ordinal models}
\label{sec:characterization}

Some known characterizations of the models presented in the previous section are due to \cite{BouyssouGrecoMatarazzoPirlotSlowinski2002a},
see also \cite{sounderpandian1991value} and \cite{segal2002min}.

\begin{theorem}[\citealp{BouyssouGrecoMatarazzoPirlotSlowinski2002a}]
  $\cgeq$ can be represented by MAX if $\cgeq$ is a weak order and the following equivalent conditions hold:
  \begin{enumerate}
  \item For all $i \in N, x_i, y_i \in X_i, a_{-i},b_{-i} \in X_{-i}$ and $w \in X$, we have
    \begin{equation}
      [x_ia_{-i} \cgeq w] \Rightarrow [y_ia_{-i} \cgeq w \text{ OR } x_ib_{-i} \cgeq w]
    \end{equation}
  \item For all $x,y \in X, i \in N$:
    \begin{equation}
      [x_iy_{-i} \cgeq x] \text{ OR } [y_ix_{-i} \cgeq x]
    \end{equation}
  \item For all $i \in N, y_i \in X_i, z_{-i} \in X_{-i}, x \in X$:
    \begin{equation}
      [y_ix_{-i} \cgt x] \Rightarrow [y_iz_{-i} \cgt x].
    \end{equation}
  \end{enumerate}
\end{theorem}

\begin{theorem}[\citealp{BouyssouGrecoMatarazzoPirlotSlowinski2002a}]
  $\cgeq$ can be represented by MIN if $\cgeq$ is a weak order and the following equivalent conditions hold:
  \begin{enumerate}
  \item For all $i \in N, x_i, y_i \in X_i, a_{-i},b_{-i} \in X_{-i}$ and $w \in X$, we have
    \begin{equation}
      [w \cgeq x_ia_{-i}] \Rightarrow [w \cgeq y_ia_{-i} \text{ OR } w \cgeq x_ib_{-i}]
    \end{equation}
  \item For all $x,y \in X, i \in N$:
    \begin{equation}
      [x \cgeq x_iy_{-i}] \text{ OR } [x \cgeq y_ix_{-i}]
    \end{equation}
  \item For all $i \in N, y_i \in X_i, z_{-i} \in X_{-i}, x \in X$:
    \begin{equation}
      [x \cgt y_ix_{-i}] \Rightarrow [x \cgt y_iz_{-i}].
    \end{equation}
  \end{enumerate}
\end{theorem}

\begin{theorem}[\citealp{BouyssouGrecoMatarazzoPirlotSlowinski2002a}]
  $\cgeq$ can be represented by $OS_{n-1}$ if $\cgeq$ is a weak order and the following equivalent conditions hold:
  \begin{enumerate}
  \item For all $i,j \in N (i \neq j), x_i, y_i \in X_i, x_j, y_j \in X_j, a_{-i} \in X_{-i}, b_{-j} \in X_{-j}, c_{-ij} \in X_{-ij}$ and $w \in X$, we have
    \begin{equation}
      [x_ia_{-i} \cgeq w \text{ AND } x_jb_{-j} \cgeq w] \Rightarrow [y_ia_{-i} \cgeq w \text{ OR } y_jb_{-j} \cgeq w \text { OR } x_{ij}c_{-ij} \cgeq w]
    \end{equation}
  \item For all $x,y \in X, i,j \in N (i \neq j)$:
    \begin{equation}
      [x_iy_{-i} \cgeq x \text{ AND } x_jy_{-j} \cgeq x] \text{ OR } [y_{ij}x_{-ij} \cgeq x]
    \end{equation}
  \item For all $x,y \in X$, all $i,j \in N (i \neq j)$, and all $z_{-ij} \in X_{-ij}$:
    \begin{equation}
      [y_ix_{-i} \cgt x \text{ AND } y_jx_j \cgt x] \Rightarrow [y_{ij}z_{-ij} \cgt x].
    \end{equation}
  \end{enumerate}
\end{theorem}

\subsection{Unified characterization of the ordinal models: $p^{{\cal A}{\cal B}}$ and subcases}
\label{sec:new-results}

Since MIN and MAX are special cases of $OS_k$, which in turn is a special case of the lattice polynomial models $p^{{\cal A}{\cal B}}$, it
is desirable to build a unified characterization for all of them. In this section we provide some steps towards such result.

\begin{theorem}
  $\cgeq$ can be represented by a lattice polynomial $p^{{\cal A}{\cal B}}$ if $\cgeq$ is a weak order, satisfies \textbf{A2}, and for any $w,x \in X$ exist
  $K \in {\cal A}, M \in {\cal B}$ with $K \cap M \neq \emptyset$, such that for any $a_{-K} \in X_{-K}$ and $b_{-M} \in X_{-M}$ we have:
  \begin{equation}
    \label{eq:a-lp}
    \left\{
        \begin{aligned}
          w \cgeq x \Rightarrow w \cgeq a_{-K}x_K, K \in {\cal A},\\
          x \cgeq w \Rightarrow b_{-M}x_M \cgeq w, M \in {\cal B}.
        \end{aligned}
      \right.
  \end{equation}
\end{theorem}

Note that, because sets ${\cal A}$ and ${\cal B}$ are finite, the axiom can also be re-written similar to the conditions in the previous
section, i.e. using ``OR'' statements. However, we feel this form is more compact. Particular cases of the above axiom include $OS_k$ and MIN/MAX.

\begin{lemma}
  $\cgeq$ can be represented by $OS_k$ if $\cgeq$ is a weak order, satisfies \textbf{A2}, and for any $w,x \in X$ there exist $K : K \subset N, \vert K \vert = k$ and
  $M : M \subset N, \vert M \vert = n - k + 1$ with $K \cap M \neq \emptyset$, such that for any $a_{-K} \in X_{-K}$ and $b_{-M} \in X_{-M}$
  we have
\begin{equation}
      \left\{
        \begin{aligned}
          w \cgeq x &\Rightarrow w \cgeq a_{-K}x_K, \\
          x \cgeq w &\Rightarrow b_{-M}x_M \cgeq w.
        \end{aligned}
      \right.
\end{equation}
\end{lemma}

\begin{lemma}
  $\cgeq$ can be represented by MIN if $\cgeq$ is a weak order and for any $w,x \in X$ exists $i \in N$, such that for
  any $a_{-i} \in X_{-i}$ we have
\begin{equation}
      \left\{
        \begin{aligned}
          w \cgeq x &\Rightarrow w \cgeq a_{-i}x_i,\\
          x \cgeq w &\Rightarrow x \cgeq w.
        \end{aligned}
      \right.
\end{equation}
\end{lemma}

\begin{lemma}
  $\cgeq$ can be represented by MAX if $\cgeq$ is a weak order and for any $w,x \in X$ exists $i \in N$, such that for any $b_{-i} \in X_{-i}$ we have
\begin{equation}
      \left\{
        \begin{aligned}
          w \cgeq x &\Rightarrow w \cgeq x,\\
          x \cgeq w &\Rightarrow b_{-i}x_i \cgeq w.
        \end{aligned}
      \right.
\end{equation}
\end{lemma}

The second condition in two last lemmas is trivial and is given only to emphasize the similarity of the axiom to the one used above. Note
also, that the first conditions in MIN/MAX characterizations are identical to those given in Section \ref{sec:characterization}.

Although the condition in two last lemmas is sufficient for characterization of MIN and MAX, in general, variations of \eqref{eq:a-lp} are
not powerful enough to characterize $p^{{\cal A}{\cal B}}$ and $OS_k$. One reason for this is that in the MIN/MAX case the axioms imply our
\textbf{A2} (the axiom that is called \textbf{AC1} in \cite{bouyssou2009conjoint}) -- in other words they imply existence of weak orders on
individual dimensions. This does not seem to be the case for the $p^{{\cal A}{\cal B}}$ and $OS_k$ conditions that we gave. Hence, we had to add
\textbf{A2} to the first two results.

\subsection{Characterization of the ordinal models in our framework}
\label{sec:relation-our-model}

In \citep{timonin2016axiomatization} we gave details of the construction of the Choquet integral for cases when every subset $X^{S_i}$ has
only one essential variable. We  now provide more details on this result.

\begin{lemma}
  \label{lm:ord-01-cap}
  Let the conditions of Theorem \ref{theo:c3-repr} hold and let there be only one essential variable on each $X^{S_a}$. Then, $\nu$ is a 0--1 capacity.
\end{lemma}
\begin{proof}
  This immediately follows by construction (see Section \ref{sec:case-with-single}). As at every $x \in X$ we have $C(\nu,x) = f_i(x_i)$,
  where $i$ is the variable essential on $X^{S_i} \ni x$, by the definition of the Choquet integral and monotonicity of $\nu$ it follows
  that $\nu$ only takes values 0 and~1.
\end{proof}

\begin{lemma}
  Let the conditions of Theorem \ref{theo:c3-repr} hold and let there be only one essential variable on each $X^{S_a}$.
  \begin{itemize}
  \item $\cgeq$ can be represented by $p^{{\cal A}{\cal B}}$;
  \item If the essential variable on every $X^{S_i}$ is the $\R$-minimal one, then $\cgeq$ can be represented by MIN;
  \item If the essential variable on every $X^{S_i}$ is the $\R$-maximal one, then $\cgeq$ can be represented by MAX;
  \item If the essential variable on every $X^{S_i}$ is the $\R$-$k$-minimal one, then $\cgeq$ can be represented by $OS-k$.
  \end{itemize}
\end{lemma}
\begin{proof}
  The first statement follows from Theorem \ref{theo:choquet-eqiv-LP}. Other follow by construction and from the uniqueness properties of
  the representation \eqref{eq:c3-repr} in the ordinal case (see Theorem \ref{theo:c3-uniqueness}). If $\S$ ordering is incomplete, then only one
  $\R$ ordering can exists which does not contradict \textbf{A3,A7} and the condition that only one variable is essential on every
  $X^{S_a}$. This follows from the uniqueness of the capacity and the uniqueness properties of the value functions.
\end{proof}

\subsection{Cardinal models}
\label{sec:cardinal-models}

The particular cases of the Choquet integral in the case of cardinal value functions are related to the convexity of the capacity. We give a
characterization of the convex capacity (the concave case  is easily obtainable by reversing the preference). Note that in the two-dimensional
case, the class of the  Choquet integrals with respect to convex capacities coincides with the class of Gilboa--Schmeidler maximin
models. In the general case of $n$ dimensions, every Choquet integral with respect to a convex capacity is a Gilboa--Schmeidler model -- the
integral is a minimum of integrals with respect to probability distributions from the capacity's core \citep{gilboa1994additive} -- but not
other way round.

To our  knowledge, this is the first result which characterizes convexity of a capacity using only the primitives of $\cgeq$ and works
in ordinal or mixed as well as purely cardinal cases, i.e. it is suitable for situations when standard sequences are not available.

\begin{theorem}
  Let conditions \textbf{A1--A9} and structural assumptions hold. Then, we have
  \begin{description}
  \item[A10 -- Convexity] For all $i,j \in N$ and for all $a_i,b_i, c_i, d_i \in X_i$, $p_j,q_j,r_j,s_j \in X_j$, and all
    $z_{-ij} \in X_{-ij}$ we have
    \begin{equation}
      \left.
        \begin{aligned}
          a_ip_jz_{-ij} & \sim b_iq_jz_{-ij}\\
          a_ir_jz_{-ij} & \sim b_is_jz_{-ij}\\
          c_ip_jz_{-ij} & \sim d_iq_jz_{-ij}  \\
          d_i & \cgeq_i c_i \\
          r_j & \cgeq_j s_j
        \end{aligned}
      \right\} \Rightarrow c_ir_jz_{-ij} \cgeq d_is_jz_{-ij},
    \end{equation}
provided $j \R i$ at $a_ip_jz_{-ij}, b_iq_jz_{-ij}, a_ir_jz_{-ij}, \sim b_is_jz_{-ij}, c_ip_jz_{-ij},d_iq_jz_{-ij} $ and $i \R j$ at
$c_ir_jz_{-ij}$ and $d_is_jz_{-ij}$,
  \end{description}
  if and only if $\nu$ is a convex capacity.
\end{theorem}

\begin{proof}
  Since conditions \textbf{A1--A9} and structural assumptions hold, there exists a Choquet integral representation of $\cgeq$. We can use it
  to prove the statement of the theorem. A capacity is convex if for all $i,j \in N, A \subset N, i \neq j$ we have \citep{chateauneuf2008some}:
  \begin{equation}
    \sum _{i,j \in B \subset A} m(B) \geq 0.
  \end{equation}

  First, let $c_ir_jz_{-ij} \clt d_is_jz_{-ij}$. We can write the conditions above using the M{\"o}bius form of the Choquet integral. All
  subsets of $N$ can be separated into four groups:
  \begin{itemize}
  \item $A: A \ni i, A \not \ni j$
  \item $A: A \ni j, A \not \ni i$
  \item $A: A \ni i, A \ni j$
  \item $A: A \not \ni i, A \not \ni j$.
  \end{itemize}
Hence, the value function for each of the points in the axiom can be written as follows. For example, for $a_ip_jz_{-ij}$ (note that we have
merged $A: A \ni i, A \not \ni j$ and $A: A \ni i, A \ni j$ groups by virtue of $j \R i$ at $a_ip_jz_{-ij}$):
    \begin{equation}
          \sum _{A \ni i} m(A) \min _{k \in A-ij} [f_i(a_i),f_k(z_k)] + \sum _{\substack{A \ni j \\ A \not \ni i}} m(A) \min _{k \in A-ij}
            [f_i(p_j),f_k(z_k)]  + \sum _{A \ni i,j} m(A) \min _{k \in A-ij} [f_k(z_k)].
    \end{equation}
Writing down all four conditions like this and after some trivial algebraic transformations which we omit in the name of readability (sum
first two conditions, add to the sum of the last two conditions and simplify), we get

\begin{dmath}
  \sum _{A \ni i,j}m(A)\left( \min _{k \in A-ij} [f_i(d_i),f_k(z_k)] - \min _{k \in A-ij} [f_i(c_i),f_k(z_k)] \right) + \sum _{A \ni
    i,j}m(A)\left( \min _{k \in A-ij} [f_j(r_j),f_k(z_k)] - \min _{k \in A-ij} [f_j(s_j),f_k(z_k)] \right) < 0.
\end{dmath}

We will show that both summands of the above expression are non-negative. Consider
\begin{equation}
  \sum _{A \ni i,j}m(A)\left( \min _{k \in A-ij} [f_i(d_i),f_k(z_k)] - \min _{k \in A-ij} [f_i(c_i),f_k(z_k)] \right).
\end{equation}
The difference $f_i(d_i),f_k(z_k)] - \min _{k \in A-ij} [f_i(c_i),f_k(z_k)$ is
\begin{itemize}
\item always non-negative, as $d_i \cgeq_i c_i$
\item maximal, when $A = \{i,j\}$
\item non-increasing as $A$ grows larger.
\end{itemize}
Note that, by convexity, $m(\{i,j\}) \geq 0$. Hence,
\begin{equation}
m(\{i,j\})\left( \min _{k \in \emptyset} [f_i(d_i),f_k(z_k)] - \min _{k \in \emptyset} [f_i(c_i),f_k(z_k)] \right)
= m(\{i,j\})\left( f_i(d_i) - f_i(c_i) \right) \geq 0.
\end{equation}
Next, find a maximal $f_{k^1}(z_{k^1}), k^1 \in N \setminus i,j$. Note that in the above expression we will only have one element
$\min _{k \in A-ij} [f_i(d_i),f_k(z_k)] - \min _{k \in A-ij} [f_i(c_i),f_k(z_k)]$ where $k^1$ is not redundant (since it's maximal). We get
\begin{equation}
  \begin{aligned}
    &m(\{i,j\})\left( f_i(d_i) - f_i(c_i) \right) + m(\{i,j,k^1\})\left( \min [f_i(d_i),f_{k^1}(z_{k^1})] - \min [f_i(c_i),f_{k^1}(z_{k^1})]
    \right)  \\
    &\geq [m(\{i,j\}) + m(\{i,j,k^1\})]\left( \min [f_i(d_i),f_{k^1}(z_{k^1})] - \min [f_i(c_i),f_{k^1}(z_{k^1})] \right) \geq 0.
  \end{aligned}
\end{equation}
The first inequality is since $m(\{i,j\}) \geq 0$ and the second is since $m(\{i,j\}) + m(\{i,j,k^1\}) \geq 0$, by convexity criterion. Now
pick the second largest $f_{k^2}(z_{k^2}), k^2 \in N \setminus i,j,k^1$. Using the same arguments we get
\begin{dmath}
    m(\{i,j\})\left( f_i(d_i) - f_i(c_i) \right) + m(\{i,j,k^1\})\left( \min [f_i(d_i),f_{k^1}(z_{k^1})] - \min [f_i(c_i),f_{k^1}(z_{k^1})]
    \right)  
    + m(\{i,j,k^2\})\left( \min [f_i(d_i),f_{k^2}(z_{k^2})] - \min [f_i(c_i),f_{k^2}(z_{k^2})] \right) + m(\{i,j,k^1,k^2\})\left( \min
      [f_i(d_i),f_{k^2}(z_{k^2})] - \min [f_i(c_i),f_{k^2}(z_{k^2})] \right)  
    \geq [m(\{i,j\}) + m(\{i,j,k^1\}) + m(\{i,j,k^2\}) + m(\{i,j,k^1,k^2\})] \left( \min [f_i(d_i),f_{k^2}(z_{k^2})] - \min [f_i(c_i),f_{k^2}(z_{k^2})] \right) \geq 0.
\end{dmath}
Continuing  like this we can add more and more elements  and eventually conclude that
\begin{equation}
  \sum _{A \ni i,j}m(A)\left( \min _{k \in A-ij} [f_i(d_i),f_k(z_k)] - \min _{k \in A-ij} [f_i(c_i),f_k(z_k)] \right) \geq 0.
\end{equation}
Similarly,
\begin{equation}
  \sum _{A \ni i,j}m(A)\left( \min _{k \in A-ij} [f_j(r_j),f_k(z_k)] - \min _{k \in A-ij} [f_j(s_j),f_k(z_k)] \right) \geq 0.
\end{equation}
Hence we have shown that the axiom necessarily holds if the capacity is convex. To show the inverse, assume that the axiom holds on
$X$. Writing down conditions of the axiom and simplifying as before, we get that everywhere on $X$ we should have
\begin{dmath} \sum _{A \ni i,j}m(A)\left( \min _{k \in A-ij} [f_i(d_i),f_k(z_k)] + \min _{k \in A-ij} [f_j(r_j),f_k(z_k)]
    \right) \geq \sum _{A \ni i,j}m(A)\left( \min _{k \in A-ij} [f_i(c_i),f_k(z_k)] - \min _{k \in A-ij} [f_j(s_j),f_k(z_k)] \right).
\end{dmath}
Assume $i,j$ interact. If this is not the case, the convexity criterion is trivially satisfied for $i,j$ as all $m(A)$ in the expression
above are 0 (see Lemma \ref{lm:null-mobius}). Assume also all variables are in the same interaction group. If this is not the case, $m(A)$
for $A$ containing variables not in the same interaction group as $i,j$ are again 0, and can be discarded.

With this assumption made, we can now pick points, such that $f_i(\cdot)$ and $f_j(\cdot)$ are the smallest value functions. Hence, the
above expression reduces to
\begin{equation}
  [f_i(d_i) + f_j(r_j)] \sum _{A \ni i,j}m(A) \geq [f_i(c_i) + f_j(s_j)] \sum _{A \ni i,j}m(A).
\end{equation}
Since $[f_i(d_i) + f_j(r_j)] \geq [f_i(c_i) + f_j(s_j)]$, we conclude that $\sum _{i,j \in A \subset N} m(A) \geq 0$.

Now pick points such that only $f_{k^1}(z_{k^1})$ is less than $f_i(\cdot)$ and $f_j(\cdot)$. We get
\begin{dmath}
  [f_i(d_i) + f_j(r_j)] \sum _{\substack{A \ni i,j \\ A \not \ni k^1}}m(A) + 2 f_{k_1}(z_{k^1})\sum _{A \ni i,j,k^1}m(A)  \geq [f_i(c_i) +
  f_j(s_j)] \sum _{\substack{A \ni i,j \\ A \not \ni k^1}}m(A) + 2 f_{k_1}(z_{k^1})\sum _{A \ni i,j,k^1}m(A),
\end{dmath}
or
\begin{equation}
  [f_i(d_i) + f_j(r_j)] \sum _{\substack{A \ni i,j \\ A \not \ni k^1}}m(A) \geq [f_i(c_i) + f_j(s_j)] \sum _{\substack{A \ni i,j \\ A \not \ni k^1}}m(A).
\end{equation}
From this we conclude that $\sum _{i,j \in A \subset N \setminus k^1} m(A) \geq 0$.

Continuing like this we can check all necessary sums for the convexity condition and for all pairs $i,j$. So, we have shown that the
capacity is convex provided the axiom holds.
\end{proof}


\section{Learning the Choquet integral}
\label{sec:learn-choq-integr}

\emph{Learning} the model means deriving model parameters from data. This step is essential in any practical application, and it is
normally performed towards at least one of  two goals: analysis of the data, by means of interpreting model parameters, or
prediction -- in other words, ``training'' the model on some dataset to use it with some other data.

It is well known that the quality of fit of a model depends on the model complexity and the available data. Learning a very complex
model using only a few data points would not achieve satisfactory results, just as using a very simple model might conceal some important
properties of a large and complicated dataset.

An important aspect of the learning process is its computational viability. Indeed, from the practical perspective, using a simpler but
faster model which is capable of delivering approximate answers in real-time fashion, might be preferable to employing a more precise but
also more expensive model which takes hours or days to be built.

In this section we look at various aspects of the Choquet integral learning and emphasize the consequences which our axiomatization results have
for this process. We start by an overview of the current learning techniques and then look at difficulties which arise when learning  the
Choquet integral model in the full generality. 

To learn the Choquet integral we need to derive two parts of the model from data:
\begin{itemize}
\item value functions $f_i: X_i \rightarrow R$, and
\item capacity $\nu$.
\end{itemize}
The following sections provide details on each of these components.

\subsection{Learning the capacity}
\label{sec:learning-capacity}

The majority of the theoretical and applied literature so far has concentrated on learning (``identification'') of the capacity \emph{only}. In
this approach, the value functions are assumed as given. Normally, for numerical coordinates $f_i(x_i) = x_i$ are taken (probably after some
rescaling). For categorical data, sometimes arbitrary numerical labels are used (see e.g. AHP), although the theoretical problems of this
approach are quite apparent.

A good review of the existing methods of capacity construction can be found in \cite{grabisch2008review}. In the majority of cases, the
learning process is based on minimization of some loss function (MSE, MAE, or similar), or on finding the extremum of some meaningful
expression, such as variance or entropy.

Typically, data is used to formulate constraints on the space of possible parameters (i.e. capacities). For example, if $x \cgeq y$, then
$\nu$ must be such that $C(\nu,f(x)) \geq C(\nu,f(y))$ (remember the value functions are considered known). Since the integral is a linear
function of the capacity, we get linear constraints. Eventually the polyhedron of all possible capacities is defined by the following data:

\begin{description}
\item[Learning set.] Pairwise preferences between elements of the ``learning set'' $X$.
\item[Criteria importance.] The most intuitive way to describe a multicriteria model qualitatively is, perhaps, to define the relative
  weights of its components. The process is semantically similar to that for additive models; however, due to non-additivity we can not rely
  only on values for singletons any more, but must also take into account all other subsets of $N$.
\item[Criteria interaction.] A more complicated type of knowledge about criteria is the character of their combined
  influence. In particular, criteria can complement each other, which is also known under the name of positive synergy,
  or else be redundant (resp. negative synergy).
\item[Veto and favour criteria.] Sometimes the model also includes criteria of an immense importance, so that the
  alternatives having low valuations on them will also inevitably receive low overall judgements. This kind of criterion
  is usually called ``veto'' in the literature. The opposite situation is having a criterion (or criteria) such that a
  high value on them automatically justifies a high overall valuation. Such elements are called ``favour''.
\item[Complexity controls.] Often it is deemed that interactions in groups larger than $k$ can be ignored to improve the computational
  properties of the model. The mechanism which allows us to achieve this is called \emph{$k$-additivity}. Most frequently, 2-additive
  capacities are used.
\end{description}

The following indices were originally applied for behavioral analysis of non-additive measures. However, they also
allow us to formulate and solve the inverse problem of capacity identification (see \citealp{marichal2000determination} and
references therein).
\begin{definition}[\citealp{shapley1953value}]
The Shapley value is an additive measure $\phi_{\nu}:2^N \rightarrow [0,1]$ defined as
  \begin{equation}
    \phi_\nu(i) = \sum \limits _{T \subset N \setminus i}
    \frac{(\vert N \vert - \vert T \vert - 1)!\vert T \vert !}{\vert N \vert !}
    [\nu(T \cup i) - \nu(T)].
\end{equation}
It can also be expressed via the M{\"o}bius transform coefficients:
\begin{equation}
\phi_m(i) = \sum \limits _{ T \subset N \setminus i }
			   \frac{1}{\vert T \vert + 1} m(T \cup i).
\end{equation}
\end{definition}

The semantic interpretation given to the Shapley value of a criterion $i \in N$ in the literature is the relative importance of the said
criterion in the decision problem. More formally, it amounts to the average marginal input of that criterion to all subsets of $N$. Being a
probability measure, the Shapley value sums up to 1 over all $i \in N$. Table \ref{tab:shapley} demonstrates how the Shapley value can be
used in capacity identification problems ($ \delta_{SH} $ is some small value -- the indifference coefficient).
\begin{table}[h!]
  \centering
  \caption{Criteria importance modelling}
\label{tab:shapley}
  \begin{tabular}{l  r }
 \toprule
    The criterion $ i $ is more important than $ j $ & $\phi_\nu(i) - \phi_\nu(j) \geqslant \delta_{SH}$ \\
\midrule
    Criteria $ i $ and $ j $ are equally important & $ -\delta_{SH} \geqslant \phi_\nu(i) - \phi_\nu(j) \leqslant \delta_{SH}$ \\
\bottomrule
  \end{tabular}\\
\end{table}
Intuition about the relative importance of a criterion can be expressed as $\phi_\nu(i) = k $ or $\phi_\nu(i) \in
[k^l, k^u] $, although, just like in the additive case, doing so is not strictly sensible.

The measure of criteria interaction character and strength was introduced by \cite{murofushi1993techniques} for pairs
of elements and later generalized by \citet{grabisch1997k}.
\begin{definition}
The interaction index of a subset $T \subset N$ is defined as
  \begin{equation}
I_\nu(T) = \sum \limits _{k=0}^{\vert N \vert -\vert T \vert} \xi_k^{\vert T \vert}
	  \sum \limits _{K \subset Z \setminus T, \vert K \vert = k}
	  \sum \limits _{L \subset T} (-1)^{\vert T \vert - \vert L \vert} \nu(L \cup K),
\end{equation}
where
\begin{equation}
\xi_k^p = \frac{(\vert N \vert - k - p)!k!}{(\vert N \vert -p+1)!}.
\end{equation}
\end{definition}
For practical problems we are particularly interested in the index expression for pairs $\{i, j\}$:
\begin{equation}
I_\nu( ij)
		= \sum \limits _{T \subset N \setminus ij}
		\xi_{\vert T \vert}^2\left[\nu(T\cup ij) - \nu(T \cup i) - \nu(T \cup j) + \nu(T)\right], 				
\end{equation}
or, when expressed with the M{\"o}bius transform coefficients:
\begin{equation}
I_m(ij)
		= \sum \limits _{T \subset N \setminus ij}
		\frac{1}{\vert T \vert + 1} m(T \cup ij).
\end{equation}

The interaction index for singletons coincides with the Shapley value. The index can be interpreted as the degree of
interaction between elements in the set $T$. Its values lie in the interval $ [-1;1] $, with 1 corresponding to the
maximal positive interaction (complementarity), and $-$1, accordingly, to the maximal negative interaction (redundancy). Table
\ref{tab:interaction} summarizes index usage in identification problems.

\begin{table}[h!]
  \centering
  \caption{Modelling criteria interactions}
\label{tab:interaction}
  \begin{tabular}{l r}
\toprule
    Criteria $ i $ and $ j $ complement each other & $ 0 \leqslant I_\nu(i,j) \leqslant 1  $ \\
\midrule
    Criteria $ i $ and $ j $ complement \\ each other stronger than $ k $ and $ l $  & $ I_\nu(i,j) - I_\nu(k,l) \geqslant \delta_{I} $ \\
\midrule
    Criteria $ i $ and $ j $ interact \\ in a way similar to $ k $ and $ l $ & $ -\delta_{I} \geqslant I_\nu(i,j) - I_\nu(k,l) \leqslant \delta_{I}  $ \\
\bottomrule
  \end{tabular}\\
\end{table}

To model ``veto'' and ``favour'' criteria we can proceed in the following way \citep{grabisch1997veto}. If some criterion $ i $ is a ``veto''
one, then
\begin{equation}
  \nu(A) = 0 \qquad \forall A \nsupseteq i.
\end{equation}
Else, if some criterion $ i $ is a ``favour'' one, then
\begin{equation}
  \nu(A) = 1 \qquad \forall A \supseteq i.
\end{equation}

Finally, if the problem allows us to employ a learning set, the DM might be asked to express his preferences with regard to
its elements. In an identification problem this corresponds to linear constraints (since the integral is linear in $\nu$)
outlined in Table \ref{tab:learning}.

\begin{table}[h!]
  \centering
  \caption{Preferences over learning set objects}
\label{tab:learning}
  \begin{tabular}{l  r }
 \toprule
    The alternative $ z_1 $ is preferred to $ z_2 $ & $C(\nu,f(z_1)) - C(\nu,f(z_2)) \geqslant \delta_{LS}$ \\
\midrule
    The DM is indifferent between $ z_1 $ and $ z_2 $  & $ -\delta_{LS} \geqslant C(\nu,f(z_1)) - C(\nu,f(z_2)) \leqslant \delta_{LS}$ \\
\bottomrule
  \end{tabular}\\
\end{table}

Having the available information expressed as a set of linear constraints we obtain the set $ \mathcal{U} $.  Summing up the results of the
previous section, $ \mathcal{U} $ can be written down as shown in equation (\ref{eq:uncset}).
\begin{figure}[h!]
  \centering
  \begin{equation}\label{eq:uncset}
    \begin{aligned}
      \mathcal{U}: \\
      \\
      &\textbf{Information from the DM} \\
      &\phi_\nu(i) - \phi_\nu(j) \geqslant \delta_{SH}, \quad  i,j \in 1,\ldots,n \\
      &\dots \\
      &-\delta_{SH} \geqslant \phi_\nu(i) - \phi_\nu(j) \leqslant \delta_{SH}, \quad  i,j \in 1,\ldots,n \\
      &\dots \\
      &I_\nu(i,j) - I_\nu(k,l) \geqslant \delta_{I}, \quad  i,j \in 1,\ldots,n  \\
      &\dots \\
      &-\delta_{I} \geqslant I_\nu(i,j) - I_\nu(k,l) \leqslant \delta_{I}, \quad  i,j \in 1,\ldots,n  \\
      &\dots \\
      &C(\nu,f(z_i)) - C(\nu,f(z_j)) \geqslant \delta_{LS}, \quad  i,j \in 1,\ldots,n  \\
      &\dots\\
      &-\delta_{LS} \geqslant C(\nu,f(z_i)) - C(\nu,f(z_j)) \leqslant \delta_{LS}, \quad  i,j \in 1,\ldots,n  \\
      &\dots\\
      & \nu(A) = 1, \forall A \supset \text{favour criteria} \\
      & \nu(A) = 0, \forall A \not\supset \text{veto criteria} \\
      &\textbf{Technical constraints} \\
      &\nu(\emptyset) = 0 \\
      &\nu(N) = 1 \\
      &\nu(B) \geq \nu(A) \quad \forall B \subset A \subset N \\
      &\textbf{Additional constraints} \\
      &k- \text{additivity. Not always applicable.} \\
    \end{aligned}
  \end{equation}
  \caption{Encoding the information as constraints on the set of capacities}
\end{figure}

Notably, all constraints are linear, and thus the set ${\cal U}$ is a polyhedron in $\mathbb{R}^{2^n}_+$.  Its dimension can be reduced to
$2^n-2$ if we exclude the $\emptyset$ and $N$ coordinates, which have fixed values. It can be reduced even further by using $k$-additive
capacities which, however, is not always possible.  By solving the feasibility problem
\begin{equation}
  \begin{aligned}
    &\min_{\nu } 1 \\
    \text{s.t. } & \nu \in \mathcal{U},
  \end{aligned}
\end{equation}
we can check if there exists at least one capacity compliant with the given data. If such capacity cannot be found, the following problem can
be solved:
\begin{equation}
\begin{aligned}
    &\min_{\nu } {\cal L}({\cal U}) \\
    \text{s.t. } & \nu \text{ is a capacity},
  \end{aligned}
\end{equation}
where ${\cal L}({\cal U})$ is some loss function of the data (e.g. the number of preference reversals). The loss function, whether an
error-based one or some other as mentioned above, is typically a convex function, so the optimization problem is quite efficient. If the
model is built for forecasting purposes, regularization techniques can also be used
\citep{tehrani2013ordinal,tehrani2012learning,tehrani2011learning,tehrani2011choquistic}. Additionally, identification problems can have
more than one solution, which induces the problem discussed below.

\subsection{Learning the value functions}
\label{sec:learn-value-funct}

Learning the value functions on the other hand is a different matter. Let us consider first how the process is performed in the additive value
function model. Recall that the model has the following form:
\begin{equation}
  x \cgeq y \iff \sum _{i=1} ^n f_i(x_i) \geq \sum _{i=1} ^n f_i(y_i).
\end{equation}
The data in such a learning problem is typically given as pairwise preferences for some points from the set $X$. The resulting problem is then
an LP, because additive value functions are linear with respect to each $f_i$ that we are aiming to learn. A well-known family of learning
methods related to learning of the additive value models are called the ``UTA methods'' \citep{siskos2005uta}. The value functions are
assumed to be linear interpolations of the learning points (i.e. they are piecewise linear), but sometimes polynomial or spline-based
versions are used \citep{sobrie2016uta}. Still, the process remains computationally efficient.

Note that the value functions learned in this manner do not provide any ``qualitative'' information about the data to the analyst. They can
be used for forecasting purposes, but due to the restrictions of the additive model, no statements about the ``importance'' of criteria or
similar notions can be made. In contrast, learning value functions and the capacity in the Choquet integral is valuable even if the value
functions are learned in a non-parametric manner. Indeed, it is the capacity that is capable of showing the qualitative relations between
criteria of the multidimensional problem, as is to some extent attested by the majority of the existing practical applications. However,
this process has two complications: the computational complexity and the confounding of the capacity and the value functions.

As mentioned above, the vast majority of the theoretical and practical contributions to the literature assume the existence of value
functions, or what is the same, of a common scale on which all attributes of the problem are being measured. This is clearly a very strong
assumption, but it also leads to a significant simplification of the learning process. Indeed, in this case we only need to learn the
capacity, which is generally a convex minimization problem. In contrast, when learning \emph{both} the capacity and the value functions, we
must solve a difficult non-convex optimization problem. Only a few papers have attempted to tackle this issue
\citep{angilella2004assessing,goujon2013holistic,angilella2015stochastic}, all of them offering some heuristic methods and small-scale
examples. This is not surprising. Indeed, consider the data point $x \cgeq y$ for some $x,y \in X$. In the Choquet integral model, it is
represented by the following expression: $C(\nu,f(x)) \geq C(\nu,f(y))$. Since the integral is a sum of products of elements of $\nu$ and
$f(x)$, the constraint is not linear in contrast to the case where only capacity is considered unknown. Moreover, it is generally
non-convex. Hence, the process of construction of the capacity and the value functions involves solving a non-convex optimization problem,
which is known to be computationally hard.

\subsection{Confounding of the capacity and the value functions}
\label{sec:conf-capac-value}

The second issue in the Choquet integral learning problems is the non-uniqueness of the resulting capacity. Even in cases where only
capacity is being learned, the exponential number of the coefficients ($2^n - 2$, excluding $\nu(\emptyset)$ and $\nu(N)$) means that the
task of model learning quickly becomes very difficult as the number of dimensions of the model increases. Typically a learning dataset which
is not sufficiently large does not allow the capacity to be learned in a precise way. This is a very well-known problem in  general
learning theory \citep{hullermeier2012vc} and it can be addressed by a number of methods. Among these we can mention the general regularization
approaches \citep{tehrani2013ordinal,tehrani2012learning,tehrani2011learning,tehrani2011choquistic}, but also some specialized methods which
can be applied when the model is used in particular applications, such as sorting
\citep{angilella2015stochastic,angilella2010non}. Additionally, a number of methods were developed for robust decision making with the
Choquet integral. Thus, in \cite{timonin2011robust} we proposed an algorithm for regret-minimizing optimization when the capacities are
only known to belong to a certain set, whereas \cite{benabbou2015minimax,benabbou2014incremental} looked at the problem of the robust capacity
construction using interactive data.

Axiomatization introduced in this work adds another level of complexity to the uniqueness problem. Indeed, the uniqueness results state that
meaningful and unique decomposition of the capacity and the value functions is only possible when the model exhibits sufficient levels of
non-separability. In particular, pairwise violation of $ij$-triple cancellation should be present to a sufficient extent to obtain a unique
capacity (in particular, all variables should be in the same interaction group, see Section \ref{sec:c3-uniqueness}). Thus, even an indefinite
amount of data, not containing a sufficiently rich structure of preferences, would lead to a strongly non-unique capacity. In fact, it is
easy to show that the capacity in such cases can be taken almost arbitrarily. Consider the extreme example, when there is no pairwise
interaction in the model. In this case, we have $n$ interaction groups of size 1 or, in other words, an additive value model. In the
expression $w_1f_1(x_1) + \cdots + w_nf_n(x_n)$ we can arbitrarily change the ``weights'' $w_i$ by compensating their increase or decrease
by a proportional change in $f_i$. The whole model can then be rescaled so that the weights sum up to 1. It is trivial that these
modifications do not affect the validity of the representations.

Non-uniqueness of the capacity is not necessarily a problem for prediction applications; however, qualitative conclusions, commonly made
based on capacity indices, become meaningless. For example, consider the paper of \cite{gang2012hotels}. Here, data from hotel evaluations on
the tripadvisor website is  analysed with the Choquet integral. Each hotel is reviewed based on several criteria, such as price, location,
etc. In addition, every hotel gets an overall mark, which allows the authors to construct the relation between general attractiveness of the
hotel and its particular features or their combinations. Reviewers are categorized into several social groups (``American businessmen'',
``European families'', etc). The paper shows which attributes and combinations of attributes are important for every group by finding
capacities that provide the best fit of the 2-additive Choquet integral to the corresponding dataset. Shapley values and interaction indices
of these capacities provide the required information.

From our perspective, the important point is that the evaluations are \emph{assumed} to be on the same scale. Every criterion is given from
one to five stars, and so is the global evaluation. Of course it seems not completely unreasonable to suggest that various incommensurable
notions such as ``5 minutes from the train station'' and ``very clean'' are somehow mapped onto a global ``satisfaction'' scale in the mind
of the reviewer, indeed there are many examples of such ``cross-modality'' mappings in the psychological literature (see Section
\ref{sec:interpretation-psy}). However, there is no real evidence supporting this claim, and we can also assume that stars on each dimension
signify just the ranking within the dimension itself and not across dimensions as the authors conjecture. The other consequence of such
assumption is that the scale is equispaced, in the sense that the (cardinal) difference between one and two stars is the same as between two
and three and between four and five. Apparently, this does not have to be true and often is not.

The possibility to fit not only the capacity but also the value functions resolves these methodological issues. Apparently, it should
also improve the quality of the fit. However, in cases when we assume a common scale, the lack of interaction between certain criteria is
not an issue -- we still obtain a unique capacity (see also axiomatizations in \citealp{wakker1989additive} and \citealp{schmeidler1989subjective}) and
corresponding indices, which would show a lack of interaction. In contrast, without the commensurability assumption, having even two
interaction groups would mean that we are not able to talk about ``criteria importance'' globally, but only within these groups. The problem
here is not with the tools used for capacity interpretation, in this case the Shapley value, but rather comes from the limitations of the
model per se. Unfortunately, it is not easy to see how this problem can be resolved, as it is in fact the same issue as the impossibility of
meaningfully using the notion of ``criteria weights'' in the additive model \citep{bouyssou2000evaluation} (Chapter 6). It is notable,
however, that the value of the interaction index would remain zero for any two elements from different interaction groups, no matter how we
transform the capacity!\footnote{See Lemma \ref{lm:null-mobius} and the definition of the interaction index given earlier in this chapter.}




\section{Interpretations and discussion}
\label{sec:discussion-results}

Motivation for this thesis came primarily from MCDA applications. However, our results can be also applied in several other subfields of decision theory. In this section
we discuss two of them -- the state-dependent utility and the social choice problems.

\subsection{Multicriteria decision analysis}
\label{sec:interpretation-maut}

MCDA provides perhaps the most natural context for our results. Indeed, in the multicriteria context the heterogeneity of the decision space
dimensions is natural and the insufficiency of the previous results is apparent and has been discussed in the literature multiple times
(e.g. \citealp{bouyssou2009conjoint}). We have covered many aspects of the Choquet integral usage in MCDA in the previous
chapters. An introduction and an example of a multicriteria model are given in Section \ref{sec:choq-integr-pref}, while questions of the model
learning and interpretation are discussed in Section \ref{sec:learn-choq-integr}, together with an example of a practical application.

From the theoretical perspective, in the multicriteria context our results imply that the decision maker constructs a mapping between the
elements of the criteria sets (their subsets to be precise). Some authors interpret this by saying that criteria elements sharing the same
utility values present the same level of ``satisfaction'' for the decision maker \citep{grabisch2008decade}. Technically, such statements
are meaningful, in the sense that permissible scale transformations do not render them ambiguous or incorrect, unless the representation is
additive. However, the substance of the statements such as ``$x_1$ on criterion 1 is at least as good as $x_2$ on criterion 2'' (which would
correspond to $f_1(x_1) \geq f_2(x_2)$) is not easy to grasp. Apart from the satisfaction interpretation, perhaps one could think about
workers performing various tasks within a single project. From the perspective of a project manager, achievements of various workers,
serving as criteria in this example, can be level-comparable despite being physically different, if the project has global milestones
(i.e. scale) which are mapped to certain personal milestones for every involved person. The novelty of our characterization is that this
scale is not given a priori. Instead, we only observe preferences of the project manager and infer all corresponding mappings from them. It
is also worth mentioning that value functions for any interacting pair can be seen to form a so-called Guttman scale (or a biorder)
\citep{guttman1944basis,doignon1984realizable}.

\subsection{Psychology}
\label{sec:interpretation-psy}

An interesting connection is that in psychology there exists a body of results on the so-called cross-modality matching. A large number of
studies have been conducted in this area since 1950s, with experiments related to loudness, colour, size, tone, pain, money, etc.
\citep{stevens1980cross,stevens1965cross,stevens1959cross,galanter1974cross,krantz1972theory}. \citet{kahneman2011thinking} gives the
following example: ``A girl learned to read when she was four. How tall is a man who is as tall as Julie was precocious?'' Normally, kids
start reading at around 5 or 6, so perhaps the girl is somewhat more precocious than average, although not by too much. Therefore, we could
say that the man is somewhat higher than the average 180\,cm, perhaps his height is 190 or similar.  Apparently our ability to answer this
question is based on the existence of some information about the distribution of the age when children start reading, and the distribution
of height. The information can come in a number of forms: either just a mean value (``on average kids start reading at 5'', ``an average man
is 180\,cm high''), or two absolute reference levels on both dimensions -- ``children start reading between 3 and 6'', ``men heights are in the
range of 165--205\,cm''. Finally, we can have  complete information about both distributions and pick a match based on that. It is this
information that allows us to ``map'' four years to something like 190\,cm. We can perhaps consider the probability of a certain value as the
universal scale shared by two distinct elements: ``75\% of children start reading after 4'', ``75\% of men are lower than 190\,cm'',
etc. However, as discussed above, such information is not always available, and there might be other mechanisms by which such mappings are
performed.


\subsection{State-dependent utility }
\label{sec:state-depend-util}

We will show how the traditional comonotonic-based axiomatization \emph{implies} state-independence and how our approach can be used to
construct a truly state-dependent model without making additional assumptions about correspondence between outcomes in different states.

The state-dependent utility concept, as introduced in Chapter \ref{chap:Title chapter 1} and further in Appendix \ref{chap:extended lit review},
is evoked when the nature of the state itself is of significance and it is not assumed that outcomes in different states have the same
meaning or value to the decision maker. A popular example is healthcare, where various outcomes can have major effects on the personal value
of the insurance premium \citep{karni1985decision}. One way to model this is to use different value functions for every state; moreover, we
could also consider the notion of state--prize \citep{karni1985decision,karni2016expected}, which actually takes the state-dependent model
directly to the heterogeneous product set case (dimensions are sets of ``state--prizes'').

So far the axiomatizations of the state-dependent utility models have been based on the existence of some correspondence between the outcomes in
different states \citep{karni2016expected,karni1993definition,karni1985decision,fishburn1973mixture}. In essence, this is not different from
assuming the homogeneous product set again, albeit with some technical differences (e.g. the decision space might only be a subset of the
full product). Although, in principle, the existence of a preference relation on the set of state--prizes is not unrealistic, it is not clear
whether this data is observable (contrary to the preferences on acts which are supposed to be always observable). Without such a relation
the additive value model (think SD-EU) does not allow us to disentangle probabilities and utilities at all (see discussion in the previous
section and earlier). The other question is whether this gives any real methodological advantage compared to using a union of state--prizes
on every dimension and proceeding as normal. A detailed discussion of this question is given in \cite{karni2016expected} and references
therein, and we do not pursue it further here. Finally, we would like to mention that the problem of state-dependence in rank-dependent
models is not well developed -- the only  paper known to the author being \cite{hong1996comonotonic}, where the authors comment on the
meaninglessness of state-dependency in the normal CEU framework, again due to the confounding issues: ``with preferences over acts as the only
empirical primitive, the factorization $\nu(A)u_A(\cdot)$ becomes meaningless. Only the product $W(x,A) = \nu(A)u_A(x)$ can be derived from
preferences''.

However, the general axiomatization of the Choquet integral presented in this thesis, is the first (to the author's best knowledge) result
where state-dependence can be derived \emph{exclusively} from the preferences over acts. This constitutes a significant difference with all
earlier results. As a side result, it is easy to show that comonotonicity-based conditions actually \emph{imply} state-independence of
preferences.

\begin{lemma}
  Let $X = Y^n$. Let conditions of the Theorem \ref{theo:c3-repr} hold. If for all $x \in X$ we have $i \E^x j$ whenever $x_i = x_j$, the
  representation is state-independent.
\end{lemma}
\begin{proof}
  Saying that $i \E j$ whenever $x_i = x_j$ in our framework amounts to saying that additive representations exist on the comonotonic
  subsets of $X$. The construction implies that $f_i(x_i) = f_j(x_j)$ whenever $x_i = x_j$. This holds for all $i,j \in X$, hence we
  can use a single utility function $U:Y \rightarrow \mathbb{R}$ for all dimensions. This constitutes state-independency.
\end{proof}

Hence, parting with the assumption that the borders between additive regions actually coincide with the borders between \emph{comonotonic}
sets, allows us to introduce state-dependency into the model and to do so solely by observing the preferences between acts. The resulting
state-dependent utility functions could be used to derive the relation on the set of state--prizes which is assumed as given in earlier
works. Note that, as previously, the meaningfulness of this relation is conditional on the violation of pairwise separability in the model,
as explained in Section \ref{sec:c3-uniqueness}. In other words, the relation might not exist between prizes of certain state pairs.

\subsection{Social choice}
\label{sec:interpr-soci-choice}

If we think of the set $N$ as of a society with $n$ agents, then $X$ is the set of all possible welfare distributions. Moreover, contrary to
the classical scenario, agents could be receiving completely different goods, for example $X_1$ might correspond to healthcare options,
whereas $X_2$ to various educational possibilities. In this case it is not a trivial task to build a correspondence between different
options across agents. Our result basically states that provided the preferences of the social planner abide by the axioms given in Section
\ref{sec:axioms-definitions}, the decisions are made as if the social planner has associated cardinal utilities with the outcomes of each
agent which are \textit{unit} and \textit{level} comparable (cardinal fully comparable or CFC in terms of
\citet{roberts1980interpersonal}). Such approach is not conventional in social choice problems, where the global (social) ordering is
usually not considered as given (there are, however some papers taking this route, e.g. \citealp{ben1997measurement}). Instead, the conditions
are normally given on individual utility functions and the ``aggregating'' functional that is used to derive the global ordering. However,
one of the important questions in social choice literature is that of the interpersonal utility comparability and whether it is justifiable
to assume it or not (e.g. \citealp{harsanyi1980cardinal}). Our results show that if the global ordering of alternatives made by the society (or the social planner) satisfy
certain conditions, it is in principle possible to have individual preferences represented by utility functions that are not only unit but
also level comparable with each other.


\section{Summary}
\label{sec:c4-summary}

We have presented extensions of our characterization for the ordinal and cardinal special cases of the Choquet integral. The ordinal models
are the well-known MIN/MAX and the order statistic, and also their generalization -- the lattice polynomial. We have shown how these can be
characterized in our framework and also related our results to the previously known axiomatizations. On the cardinal side of things, we have
shown how it is possible to characterize the Choquet integral with respect to a convex capacity. The axiom is similar to the tradeoff
consistency condition and is the first characterization of convex models which can deal with both cardinal and ordinal cases (or a mixture
of the two).

Next, we  discussed various aspects of the Choquet integral learning. Traditionally, the learning of the integral was confined to
capacity learning only. However, this approach suffers from serious methodological difficulties. Namely, it requires a very strong
assumption that all criteria are measured on the same scale. We looked at how various preferential information could be used in the capacity
identification problem and analysed why the process of capacity identification is relatively computationally effective. In contrast,
learning the capacity and the value functions together seems to be computationally very hard. There have been only a few attempts at solving it
in the literature, all of them offering only some heuristic methods. Finally, we look at the problem of confounding of the value functions
and the capacity. Our characterization results state that a unique decoupling of the capacity and the value functions is possible only  when the dimensions of the decision space exhibit sufficient pairwise interaction. This has a profound impact on the learning
properties of the Choquet integral, since it guarantees that it is impossible to obtain a unique capacity if the variables are not
interacting enough, no matter how much data we have. This means that the usage of the well-known indices such as the Shapley index is
limited. An alternative option is to use the ``sum of Choquet'' representation \eqref{eq:add-decomposition}, whereby the indices become
meaningful within each interaction group.

Finally, we have looked at various interpretations of our results and their applications in decision theory. We started with MCDA, which was
the main inspiration for our research. Our axiomatization is a long-missing result in this area and we hope that it will help  promote
further theoretical research of the Choquet integral in MCDA. The characterization leads to construction of a unique mapping between
elements of various criteria sets (dimensions of the decision space). This has interesting connections to the question of cross-modality
mapping, which has been extensively studied in psychology since the 1950s. Finally, we discussed two other areas where our results can be
applied -- the social choice theory and the state-dependent DUU. The latter is especially interesting, as our characterization is the first
to construct a meaningful state-dependent model based solely on the preferences among acts. Previous works introduced additional
preference relations into the model, in particular the relation on the set of ``state--prizes''. Conceptually, this amounts to saying that
elements of various dimensions are commensurate which does not always have to be the case. Observability of this preference relation is also
not apparent. Our results do not require any additional constructs apart from the preference between acts themselves. Yet, we are able to
construct a unique mapping between the outcomes in different states (provided the data exhibits sufficient interaction).


\bibliographystyle{abbrvnat}
\bibliography{cite_lib}

\end{document}